\documentclass[a4paper,11pt]{article}

\usepackage{graphicx}

\usepackage{fullpage}
\usepackage{libertine}
\usepackage{color}

\usepackage{pdflscape}
\usepackage{afterpage}
\usepackage{capt-of}

\usepackage{amsmath,amsfonts,amssymb,amsthm}

\usepackage[breaklinks=true]{hyperref}
\usepackage[table,svgnames]{xcolor}
\usepackage[capitalise,nameinlink]{cleveref}
\hypersetup{colorlinks={true},linkcolor={DarkBlue},citecolor=[named]{DarkGreen}}

\usepackage{natbib}
\usepackage{authblk}

\usepackage[ruled,vlined]{algorithm2e}
\SetArgSty{textrm}

\newtheorem{theorem}{Theorem}[section]
\newtheorem{corollary}[theorem]{Corollary}

\newtheorem{lemma}[theorem]{Lemma}

\theoremstyle{definition}
\newtheorem*{comment*}{Comment}

\newcommand{\SW}{\text{SW}}

\begin{document}

\allowdisplaybreaks

\title{\bf Revisiting the Distortion of Distributed Voting}

\author[1]{Aris Filos-Ratsikas}
\author[2]{Alexandros A. Voudouris}

\affil[1]{School of Informatics, University of Edinburgh, UK}
\affil[2]{School of Computer Science and Electronic Engineering, University of Essex, UK}

\renewcommand\Authands{ and }
\date{}

\maketitle

\begin{abstract}
We consider a setting with agents that have preferences over alternatives and are partitioned into disjoint districts. The goal is to choose one alternative as the winner using a mechanism which first decides a representative alternative for each district based on a local election with the agents therein as participants, and then chooses one of the district representatives as the winner. Previous work showed bounds on the distortion of a specific class of deterministic plurality-based mechanisms depending on the available information about the preferences of the agents in the districts. In this paper, we first consider the whole class of deterministic mechanisms and show asymptotically tight bounds on their distortion. We then initiate the study of the distortion of randomized mechanisms in distributed voting and show bounds based on several informational assumptions, which in many cases turn out to be tight. Finally, we also experimentally compare the distortion of many different mechanisms of interest using synthetic and real-world data. 
\end{abstract}

\section{Introduction}
{\em Voting} is a ubiquitous method for making decisions with a large number of applications, such as electing political representatives, deciding how to split a public budget between projects, or choosing which services (restaurants, hotels, etc) to recommend to new users based on past user experiences. As such, it has been at the epicenter of research within multiple disciplines including political sciences, economics and computer science~\citep{comsoc-book}. The most prominent question in this research agenda is to identify the best voting rule to use to collectively aggregate the preferences of agents over alternative options into a single winning alternative, with most of the earlier literature focusing on axiomatic properties that good voting rules should have. An alternative way to tackle this question  that has been proposed in computer science is through the {\em distortion} framework~\citep{survey2021} which allows to compare different voting rules based on how well they approximate the optimal choice as measured in terms of a social objective function like the {\em utilitarian social welfare}.  

Since its inception in 2006 by \citet{procaccia2006distortion}, the distortion framework has been applied to several utilitarian social choice settings~(e.g., \citep{boutilier2015optimal,anshelevich2018approximating,gkatzelis2020resolving}). 
The lion's share of previous work has focused on centralized models with a single pool of agents whose preferences are directly given as input to a voting rule, which thus can utilize all the given information at once to make a decision. However, there are many applications with multiple pools of agents which make independent local decisions that can be thought of as recommendations for the final decision. To give a concrete example, in most political election systems, the citizens are partitioned into districts based on geographic or other criteria, and vote within their districts to propose the candidate (party) they would like to be selected as the winner. 

Inspired by situations like the one described above, \citet{distributed2020} initiated the study of the distortion of mechanisms in a {\em distributed} single-winner setting where a set of $n$ agents with \emph{cardinal preferences} over a set of $m$ alternatives are partitioned into $k$ disjoint districts. The authors focused on \emph{deterministic} mechanisms of the form {\sc Plurality-of-$f$},  which first choose a representative alternative for each district according to some rule $f$, by holding a local election with the agents of the district as the voters, and then picking the winner to be the alternative that is representative of the most districts (i.e., using the {\sc Plurality} rule). \citeauthor{distributed2020} considered mechanisms for which the rule $f$ can be \emph{cardinal} or \emph{ordinal}, i.e., it may use the actual numerical information about the preferences of the agents within the districts or just consistent rankings. The authors showed that, when the districts are symmetric (that is, each of them contains the same number of agents), the distortion of a cardinal mechanism, namely {\sc Plurality-of-Range-Voting} is $O(km)$, and provided an asymptotically matching lower bound of $\Omega(km)$ on the distortion of any {\sc Plurality-of-$f$} mechanism. For ordinal mechanisms, they showed that {\sc Plurality-of-Plurality} achieves a distortion of $O(km^2)$, and that this is asymptotically best among all ordinal {\sc Plurality-of-$f$} mechanisms.

\subsection{Revisiting the distortion of distributed voting}
A first observation about the results of \citet{distributed2020} is that there is a-priori no reason to restrict our attention to only mechanisms in the class {\sc Plurality-of-$f$}, as using other over-districts rules could potentially lead to better distortion. Indeed, follow-up work considered distributed social choice settings with metric preferences~\citep{anshelevich2022distortion,filos2021approximate} without such restrictions on the over-districts rule. In addition, all of the previous work on these settings only considered deterministic mechanisms that use deterministic in-district and over-districts rules. Randomization has proven out to be a very useful tool in achieving better (expected) distortion bounds in the centralized setting (see \cite{boutilier2015optimal, ebadian2022optimized}), so it is only natural to consider randomized mechanisms in the distributed setting as well. Finally, an important question is how the distortion bounds are affected in case the participants act selfishly, and whether there are \emph{strategyproof} mechanisms with good distortion bounds. This question has been considered in the centralized setting \citep{filos2014truthful,bhaskar2018cardinal,bhaskar2018truthful-multiwinner,ebadian2022optimized} and also in the distributed metric setting~\citep{filos2021approximate}; we consider it in the context of the normalized setting of \citet{distributed2020} as well. 

\subsection{Our Contributions}
We consider the class of \emph{all} mechanisms for distributed voting in the setting of \citep{distributed2020}. In particular, we consider the $f_{\text{over}}$-of-$f_{\text{in}}$ class of mechanisms, where $f_{\text{in}}$ is an in-district rule that takes as input the preferences of the agents within each district and outputs a representative alternative for the district, while $f_{\text{over}}$ is a rule that takes as input the representative alternatives of all districts and chooses one of them as the overall winner. We consider several different cases depending on the nature of $f_{\text{over}}$ and $f_{\text{in}}$ (deterministic or randomized), and the type of information they can utilize (cardinal or ordinal). We show the following results; see Table~\ref{tab:results} for an overview.

\begin{table}[t]
\centering
\begin{tabular}{|l|c|c|c|}
\hline
              & Deterministic          & Randomized-of-Deterministic   & Randomized-of-Randomized                            \\ \hline
Ordinal       & $\Theta(km^2)$ & $\Theta(km^2)$ & $\Omega(\sqrt{m}), O(\sqrt{m \log m})$ \\
Cardinal      & $\Theta(km)$   & $\Theta(k)$    & $\Theta(k)$                              \\ 
Strategyproof & $\Theta(nm)$   & $\Theta(nm)$   & $\Omega(\sqrt{m}), O(\sqrt{m \log m})$ \\ \hline
\end{tabular}
\caption{An overview of our results. Specific details can be found in the appropriate sections.}
\label{tab:results}
\end{table}

\smallskip
\noindent 
{\bf Deterministic Mechanisms.} When $f_{\text{over}}$ and $f_{\text{in}}$ are both deterministic and the districts are symmetric, we show that the best possible distortion is $\Theta(km)$ when the valuation functions of the agents are accessible (cardinal mechanisms), and is $\Theta(km^2)$ when only ordinal information about the preferences of the agents is available (ordinal mechanisms). The upper bounds were shown by \cite{distributed2020} and here we provide asymptotically tight lower bounds. These results show that for general, unstructured (normalized) valuations, employing different over-district rules in fact does \emph{not} result in improvements on the distortion. We present these results in Section~\ref{sec:deterministic}.
    
\smallskip
\noindent 
{\bf Randomized Mechanisms.} 
In Section~\ref{sec:randomized}, we consider for the first time the distortion of randomized mechanisms in distributed voting. 
We first prove a simple composition theorem, which shows that using an in-district rule with known distortion $\delta$ in the centralized setting and then selecting the winner uniformly at random from the set of representatives, defines a distributed mechanism with distortion $O(k \delta)$. Using this, complemented with new lower bounds, we show that the best possible distortion for cardinal unanimous mechanisms is $\Theta(k)$; in fact, this is true even when the districts are asymmetric and when $f_{\text{over}}$ is randomized but $f_{\text{in}}$ is deterministic.
    
For ordinal mechanisms, we consider two cases: 
(a) mechanisms that use deterministic in-district rules $f_{\text{in}}$, and 
(b) fully-randomized mechanisms, where both $f_{\text{over}}$ and $f_{\text{in}}$ are randomized rules. 
For (a), we show that the best possible distortion is $\Theta(km^2)$. The upper bound follows from the bound on {\sc Plurality-of-Plurality} proven in \citep{distributed2020}; here, we provide an asymptotically matching lower bound assuming a natural universal tie-breaking rule. 
For (b), we prove a simple but very interesting result: For a well-studied class of randomized centralized voting rules called \emph{point-voting schemes} (e.g., see \cite{gibbard1977manipulation, Barbera1978}), there exists a distributed implementation so that there is no effect on the induced probability distribution, even for asymmetric districts. Simply put, using such rules it is possible to escape the ill effects of districts in terms of the distortion, even when the districts are asymmetric. From this result, it follows that there exists a distributed implementation of a well-known mechanism of \citet{boutilier2015optimal} that achieves distortion $O(\sqrt{m\log{m}})$, almost matching the best possible lower bound of $\Omega(\sqrt{m})$. 
    
\smallskip
\noindent 
{\bf Strategyproof Mechanisms.}
For \emph{strategyproof} mechanisms, which are resilient to strategic manipulation, we show that a best-possible distortion of $\Theta(nm)$ for deterministic mechanisms (and more generally mechanisms with a deterministic in-district rule) is easy to achieve by a variation of a dictatorship rule. For randomized mechanisms, since point-voting schemes are strategyproof, the bound $O(\sqrt{m\log{m}})$ carries over to this class as well. Results about deterministic strategyproof mechanisms are presented in Section~\ref{sec:deterministic}, and about randomized strategyproof mechanisms in Section~\ref{sec:randomized}.

\smallskip
\noindent 
{\bf Experiments.}
Finally, in Section~\ref{sec:experiments}, we perform experiments using real-world data and synthetic data to evaluate the effect of distributed decision making to the distortion in settings closer to practice. The main conclusions of our experimental results mirror that of our theoretical results in Sections~\ref{sec:deterministic} and \ref{sec:randomized}.

\subsection{Further Related Work}
The distortion literature is by now rather extensive, including topics such as single-winner voting \citep{boutilier2015optimal,anshelevich2018approximating,gkatzelis2020resolving,kempe22veto}, multi-winner voting~\citep{caragiannis2017subset,CSV22}, matching problems~\citep{Aris14,amanatidis2021matching}, and participatory budgeting~\citep{benade2017participatory}. 
Generally speaking, most works can be categorized as either studying a normalized utilitarian setting (e.g., \citep{procaccia2006distortion,boutilier2015optimal,Aris14,benade2017participatory,ebadian2022optimized}) 
or a metric preference setting (e.g., \citep{anshelevich2016blind,anshelevich2018approximating,gkatzelis2020resolving,CSV22,kempe22veto}). 
Some more recent works have also studied the interplay between information and distortion \citep{Amanatidis2020peeking,amanatidis2021matching,amanatidis2022don,mandal2019thrifty,mandal2020optimal,abramowitz2019awareness}, and there have also been several works on strategyproofness in the context of distortion~\citep{filos2014truthful,Aris14,bhaskar2018cardinal,bhaskar2018truthful-multiwinner,ebadian2022optimized}. 
We refer the reader to the survey of \citet{survey2021} for a detailed overview of the related literature.

Besides the aforementioned works on distributed voting, \citet{borodin2019primaries} studied a related two-stage setting in which the voters participate in a central election, but the candidates themselves come from local elections within the political parties' electorates. Beyond distortion, in the context of district-based elections, there have also been other works that have considered the degree of deviation from proportional representation (e.g., see \citep{bachrach2016misrepresentation} and references therein), and some works that have studied the complexity of manipulation (e.g., see \citep{elkind2021protecting,lewenberg2017divide,lev2019reverse,borodin2018big}).

\section{Preliminaries}\label{sec:prelims}
An instance $I$ of our problem is given by a tuple $I = (N,A,\mathbf{v},D)$.
There is a set $N$ of $n$ {\em agents} (or {\em voters}) that have preferences over a set $A$ of $m$ {\em alternatives} (or {\em candidates}). The preferences of each agent $i \in N$ are captured by a {\em valuation function} $v_i:A\rightarrow \mathbb{R}_{\geq 0}$ that maps every alternative $a \in A$ to a real non-negative value $v_i(a) = v_{ia}$. Following previous work, we assume that the valuation functions are normalized such that $\sum_{a \in A} v_{ia} = 1$ for every $i \in N$ ({\em unit-sum assumption}). Let $\mathbf{v} = (v_i)_{i \in N}$ be the {\em valuation profile} consisting of the valuation functions of all agents. The agents are also partitioned into a set $D$ of $k$ disjoint {\em districts}. 

For every district $d \in D$, let $N_d$ be the set of agents it contains, such that $\bigcup_{d \in D}N_d = N$. 
In the {\em symmetric} case, each district $d$ contains exactly $\lambda = n/k$ agents. 
In the {\em asymmetric} case, each district $d$ contains a number $n_d$ of agents. 
All our lower bounds follow by instances consisting of symmetric districts, whereas our upper bounds in Section~\ref{sec:randomized} hold for asymmetric districts. 

\subsection{Mechanisms}
Our goal is to choose an alternative to satisfy several criteria of interest. This choice must be done using a {\em distributed mechanism} that uses an {\em in-district} voting rule $f_{\text{in}}$ and an {\em over-districts} voting rule $f_{\text{over}}$ to implement the following two independent steps:
\begin{itemize}
    \item Step 1: For each district $d$, choose a {\em representative alternative} $a_d \in A$ by holding a {\em local election} based on $f_{\text{in}}$.
    \item Step 2: Choose a district representative as the winner based on $f_{\text{over}}$ by considering the districts as voters and their representatives as the candidates they approve. 
\end{itemize}
For simplicity we refer to such mechanisms as {\em $f_{\text{over}}$-of-$f_{\text{in}}$}.  
Different choices of $f_{\text{in}}$ and $f_{\text{over}}$ lead to different distributed mechanisms.
Note that the in-district rule can in general use various types of information about the preferences of the agents. For instance, it may be able to use exact {\em cardinal} information about the valuation functions, or only {\em ordinal} information that is induced by the values (i.e., rankings of alternatives that are consistent to the values of the agents for them). In the latter case, we will use $\sigma_i$ to denote the preference ranking of agent $i \in N$ so that $\sigma_i(a)$ is the rank of alternative $a\in A$ in the ranking of $i$, and $\sigma_i(a) < \sigma_i(b)$ if $v_i(a) \geq v_i(b)$; let $\boldsymbol{\sigma}=(\sigma_i)_{i \in N}$ be the {\em ordinal profile} consisting of the preference rankings of all agents. To be concise in the definitions below, let $\boldsymbol{\delta}(I)$ be the information about the preferences of the agents in instance $I = (N,A,\mathbf{v},D)$ that is used by a mechanism; that is, $\boldsymbol{\delta}(I) = \mathbf{v}$ in case of cardinal information, or $\boldsymbol{\delta}(I) = \boldsymbol{\sigma}$ in case of ordinal information.

We will focus on different classes of distributed mechanisms depending on the available information about the preferences of the agents at the district level (cardinal or ordinal), and also on whether their decision is deterministic or randomized (that is, they choose the district representatives or final winner based on probability distributions). 

\subsection{Social Welfare and Distortion}
Given an instance $I$, the {\em social welfare} of an alternative $a \in A$ is the total value that the agents have for $a$, that is, $\SW(a|I) = \sum_{i \in N} v_{ia}$. So, the expected social welfare achieved by a randomized distributed mechanism $M$ that chooses alternative $a \in A$ as the winner $w$ with probability $\Pr_M[w=a]$ is 
$$\mathbb{E}[\SW(M(I))] = \sum_{a \in A} \Pr_M[w=a] \cdot \SW(a|I).$$
The efficiency of a distributed mechanism is measured by the notion of {\em distortion}. The distortion of a distributed mechanism $M$ is the worst-case ratio (over all possible instances with $n$ agents, $m$ alternatives, and $k$ districts) of the maximum social welfare achieved by any alternative over the (expected) social welfare of the alternative chosen by the mechanism as the winner $w$, that is,
\begin{align*}
    \text{dist}(M) = \sup_{I} \frac{\max_{a \in A}\SW(a|I)}{\mathbb{E}[\SW(M(\boldsymbol{\delta}(I))]}.
\end{align*}
Clearly, $\text{dist}(M) \geq 1$. When the denominator in the definition of the distortion tends to $0$, we will say that the distortion is infinite or unbounded. Our goal is to identify the best possible distributed mechanisms in terms of distortion.

\subsection{Strategyproofness}
Another important property that we would like our mechanisms to satisfy is that of strategyproofness. A {\em strategyproof} mechanism makes decisions such that providing false information never leads to the selection of an alternative that an agent prefers over the alternative chosen when the agent provides truthful information. In particular, for any instance $I$, it must be the case that
$v_i(M(\boldsymbol{\delta}(I))) \geq v_i(M(\boldsymbol{\delta}(I')))$ for any agent $i \in N$, where $I'$ is the instance obtained when only agent $i$ reports information different than that in $I$.

\subsection{Some useful observations and properties}\label{sec:observations}
Before we present our technical results, let us briefly discuss some useful properties. 

\smallskip 
\noindent 
{\bf Locality of distributed mechanisms:} First, observe that any distributed mechanism $f_{\text{over}}$-of-$f_{\text{in}}$  satisfies a {\em locality} property in the following sense. 
A district $d$ (that is, the preferences of a number of agents) appears in different instances if in each of these instances there is a district with the same number of agents and the same information about theirs preferences as in $d$ (depending on what is required by the mechanism). Since the information is the same, the in-district rule $f_{\text{in}}$ must decide the same alternative as the representative of the district in all these instances. Similarly, in all instances where the mechanism has decided the same set of district representatives, the over-districts rule $f_{\text{over}}$ must decide the same final winner. 

\smallskip 
\noindent 
{\bf Distortion of distributed vs centralized:} Another useful observation is that the distortion of a distributed mechanism $f_{\text{over}}$-of-$f_{\text{in}}$ is at least as much as the distortion of the in-district centralized voting rule $f_{\text{in}}$. Indeed, when $k=1$, there is only one representative alternative chosen by $f_{\text{in}}$, and thus this alternative must be chosen as the winner by $f_{\text{over}}$; this is also true for instances with $k \geq 2$ districts which are all copies of one district. Consequently, the distortion of $f_{\text{in}}$ is a lower bound on the distortion of $f_{\text{over}}$-of-$f_{\text{in}}$.

\smallskip 
\noindent 
{\bf Strategyproofness:} Observe that for a distributed mechanism $f_{\text{over}}$-of-$f_{\text{in}}$ to be strategyproof it is necessary that the in-district rule $f_{\text{in}}$ is strategyproof. This again follows by how the mechanism would work in instances with a single district, in which case the over-districts rule $f_{\text{over}}$ does not play any role in the selection of the final winner. 

\smallskip 
\noindent 
{\bf Unanimity:} A few of our results will require the in-district rules $f_{in}$ to be unanimous. Unanimity stipulates that if all of the agents have the same alternative as the top preference, that alternative must be selected (with probability $1$). Unanimity is a very natural property of ``reasonable'' voting rules, especially deterministic ones. For randomized rules, there might be reasons to consider non-unanimous choices, e.g., see \cite{gibbard1977manipulation,filos2014truthful}.

\section{Deterministic mechanisms}\label{sec:deterministic}
We start with deterministic distributed mechanisms and focus explicitly on the case of symmetric districts in this section (that is, the size of each district is $\lambda$). 
When full information about the valuations of the agents is known at the district level, \citet{distributed2020} showed that the mechanism {\sc Plurality-of-Range-Voting}, which chooses the representative of each district to be the alternative with maximum social welfare for the agents in the district, has distortion $O(km)$. We show that this mechanism is asymptotically best possible over all possible deterministic distributed mechanisms that use unanimous in-district rules (but may not use {\sc Plurality} as the over-districts rule).  

\begin{theorem}\label{thm:dist-unanimous}
The distortion of any deterministic distributed mechanism with a unanimous in-district rule is $\Omega(km)$.
\end{theorem}

\begin{proof}
Let $M$ be some deterministic distributed mechanism with a unanimous in-district rule. 
Without loss of generality, whenever there are $k$ distinct district representatives $\{a_1,\ldots,a_k\}$, we assume that $M$ chooses $a_1$ as the overall winner. Let $\varepsilon > 0$ be some positive infinitesimal and consider the following instance with $k$ districts $\{d_1, \ldots, d_k\}$ and $m > k$ alternatives:
\begin{itemize}
    \item In district $d_1$, all agents have value $1/m + \varepsilon$ for alternative $a_1$, and value $1/m - \varepsilon/(m-1)$ for any other alternative.
    \item For any $\ell \in \{2,\ldots,k\}$, in district $d_\ell$, all agents have value $1/2 + \varepsilon$ for alternative $a_\ell$, value $1/2 -\varepsilon$ for alternative $x$, and value $0$ for any other alternative. 
\end{itemize}
Since the in-district rule is unanimous, the district representatives are alternatives $\{a_1, \ldots, a_k\}$, and the overall winner is thus $a_1$. The social welfare of alternative $a_1$ is approximately $\lambda/m$, whereas the social welfare of alternative $x$ is approximately $k\cdot\lambda/2$, leading to distortion $\Omega(km)$.
\end{proof}

When only ordinal information about the preferences of the agents is available, \citet{distributed2020} showed that {\sc Plurality-of-Plurality}, which chooses the favorite alternative of most of the agents in a district as its representative and then the alternative that represents the most districts as the winner, has distortion $O(km^2)$. We show that this mechanism is asymptotically best possible among all ordinal distributed mechanisms (without any restrictions), thus improving upon the result of \citet{distributed2020} who showed that {\sc Plurality-of-Plurality} is best possible only within the class of mechanisms they studied. 

We first prove an easy but important lemma showing that when only ordinal information is available, to achieve finite distortion, it is necessary the representative of each district to be some alternative that is the favorite of at least one agent in the district. 

\begin{lemma}\label{lem:top-or-infinite}
The representative of any district must be some top-ranked alternative, otherwise the distortion is infinite. 
\end{lemma}

\begin{proof}
Let $d$ be a district and let $T$ be the set of top-ranked alternatives. Suppose that the representative of $d$ is chosen to be some alternative $x \not\in T$. Then, in any instance consisting of copies of $d$, the winner must be $x$. However, the valuation profile might be such that all agents have value $1$ for their favorite alternative and $0$ for any other alternative. Consequently, the social welfare of $x$ might be $0$, whereas the social welfare of any top-ranked alternative is positive, leading to infinite distortion. 
\end{proof}

We say that a district is {\em divided} if its $\lambda$ agents are partitioned into $m/2$ equal-sized sets such that all the $2\lambda/m$ agents in each set rank the same alternative first and different sets of agents have different top-ranked alternatives. By Lemma~\ref{lem:top-or-infinite}, the representative of such a district must be one of the top-ranked alternatives. The following lemma shows that choosing the representative of a divided district as the winner is, under some circumstances, a bad choice. 

\begin{lemma}\label{lem:divided-districts}
Suppose that some alternative $y_1$ is chosen as the winner by a deterministic ordinal distributed mechanism when the set of representatives is $\{y_1, \ldots, y_k\}$. If there exists a divided district that is represented by $y_1$, then there are $k-1$ districts with representatives $y_2,\ldots,y_k$, and altogether these $k$ districts define an instance such that the distortion of the mechanism is $\Omega(k m^2)$.
\end{lemma}

\begin{proof}
Let $M$ be a deterministic ordinal distributed mechanism that selects $y_1$ as the winner when the set of representatives is $\{y_1, \ldots, y_k\}$, and let $d$ be the divided district that is represented by $y_1$.
Consider the following $k$ districts:
\begin{itemize}
    \item The first district is a copy of $d$. 
    \item For every $\ell \in \{2,\ldots,k\}$, the $\ell$-th district is such that all agents therein rank $y_\ell$ first, $x \not\in \{y_1,\ldots,y_k\}$ second, and then all other alternatives. By Lemma~\ref{lem:top-or-infinite}, $M$ must choose $y_\ell$ as the representative of the $\ell$-th district, as this is the only top-ranked alternative. 
\end{itemize}
So, indeed the set of representatives is $\{y_1,\ldots,y_k\}$ and $M$ chooses $y_1$ as the winner by assumption.
One possible valuation profile is the following:
\begin{itemize}
    \item In the first, divided district, the $2\lambda/m$ agents that rank $y_1$ first have value $1/m$ for all alternatives, and the remaining agents all have value $1$ for their favorite alternative.
    \item For every $\ell \in \{2,\ldots,k\}$, all agents in the $\ell$-th district have value $1/2$ for their two favorite alternatives ($y_\ell$ and $x$).
\end{itemize}
Consequently, the social welfare of $y_1$ is $\lambda/m^2$ whereas the social welfare of $x$ is approximately $k\cdot \lambda/2$, and thus the distortion is $\Omega(k m^2)$.
\end{proof}

Lemma~\ref{lem:divided-districts} shows that deterministic ordinal distributed mechanisms with distortion $o(km^2)$ must not output the representative of a divided district as the winner when it is given a set of districts with different representatives. However, as we show in the proof of the next theorem, there are instances where such a choice is inevitable, and thus the distortion is $\Omega(km^2)$.

\begin{theorem}
The distortion of any deterministic ordinal distributed mechanism is $\Omega(km^2)$.
\end{theorem}

\begin{proof}
Let $M$ be a deterministic ordinal distributed mechanism.
We focus on instances with $k$ districts and sets of alternatives $A \cup B \cup C \cup \{x\}$, where $A = \{a_1,\ldots,a_k\}$, $B = \{b_1, \ldots, b_{m/2+k-1}\}$, and $C = \{c_1,\ldots,c_{m-2k}\}$. Without loss of generality, suppose that when the district representatives are $\{a_1,\ldots,a_k\}$, $M$ chooses $a_1$ as the overall winner.

Let $d_1$ be a divided district with set of top-ranked alternatives $\{a_1,b_1,\ldots,b_{m/2-1}\}$. By Lemma~\ref{lem:divided-districts}, if $a_1$ is the representative of $d_1$, then there exists an instance such that the distortion of $M$ is $\Omega(km^2)$. So, suppose that the representative of $d_1$ is some other top-ranked alternative, say $b_1$. 
Again by Lemma~\ref{lem:divided-districts}, if $b_1$ is chosen as the winner whenever she is part of a representative set consisting of $k$ distinct alternatives, then the distortion of $M$ would be $\Omega(km^2)$. So, let us assume that when the district representatives are $\{b_1,a_2,\ldots,a_k\}$, the winner is an alternative different than $b_1$, say $a_2$. 

We can now repeat this argument step by step for each alternative $a_\ell$, $\ell \in \{2,\ldots,k\}$. In particular, let $d_\ell$ be a divided district with top-ranked alternatives $\{a_\ell,b_\ell,\ldots,b_{m/2+\ell-2}\}$ (note that alternatives $b_1,\ldots,b_{\ell-1}$ do not appear as top-ranked alternatives in $d_\ell$). By Lemma~\ref{lem:divided-districts}, if $a_\ell$ is the representative of $d_\ell$ then the distortion of $M$ is $\Omega(km^2)$, so the representative is some other alternative from the set $\{b_\ell,\ldots,b_{m/2+\ell-2}\}$, say $b_\ell$. Again by Lemma~\ref{lem:divided-districts}, if $b_\ell$ is chosen as the winner whenever she is part of a representative set consisting of $k$ distinct alternatives, then the distortion of $M$ would be $\Omega(km^2)$. So, when the district representatives are $\{b_1,\ldots,b_\ell, a_{\ell+1},\ldots,a_k\}$, the winner is an alternative not in $\{b_1,\ldots,b_\ell\}$, say $a_\ell$. 

The last step of this repeated argument leads to the lower bound of $\Omega(km^2)$: We have reached an instance with set of representatives $\{b_1,\ldots,b_k\}$ all of whom are representative of some divided district, and thus no matter who of them is chosen as the winner, by Lemma~\ref{lem:divided-districts} there exists an instance that includes the corresponding divided district and $k-1$ unanimous districts (like in the proof of the lemma) such that the distortion is  $\Omega(km^2)$.
\end{proof}

Finally, let us discuss the case of deterministic strategyproof distributed mechanisms. \citet{bhaskar2018cardinal} showed that the distortion of any deterministic centralized strategyproof voting rule (including those that have access to the valuation functions) is $\Theta(nm)$. From the discussion Section~\ref{sec:observations}, we directly obtain a lower bound of $\Omega(nm)$ for the distributed setting as well. A tight upper bound is also not hard to derive by considering the straightforward {\sc First-of-First} mechanism which works as follows:
\begin{itemize}
    \item For each district $d$, choose the favorite alternative of the first agent therein as the representative.
    \item Choose the representative of the first district as the winner.
\end{itemize}

\begin{theorem}
{\sc First-of-First} is strategyproof and achieves an asymptotically best possible distortion of $\Theta(nm)$ within the class of deterministic strategyproof distributed mechanisms.
\end{theorem}

\begin{proof}
The mechanism is clearly strategyproof since the winner is the favorite alternative of the first agent of the first district who acts as a dictator. Since the winner is ranked first by an agent, the social welfare of the mechanism is at least $1/m$. The maximum possible social welfare is $n$, and thus the distortion is $O(nm)$.
\end{proof}

\section{Randomized mechanisms}\label{sec:randomized}
We start our discussion on randomized distributed mechanisms by analyzing a general class of mechanisms that we call {\sc Uniform-of-$\delta$-Approximate}. A mechanism $M$ in this class works as follows: 
\begin{itemize}
    \item For each district $d$, $M$ chooses the representative $a_d$ according to some centralized voting rule $f_{\text{in}}$ that has distortion at most $\delta$.
    \item $M$ chooses the winner uniformly at random from the set of representatives. 
\end{itemize}
Picking the winner uniformly at random from the representatives that have been selected seems to be the most natural choice as there is not much information about the preferences of the agents in the districts, and essentially all we can do is assign higher proportional probability to an alternative that is representative of more districts. We have the following result. 

\begin{theorem} \label{thm:uniform-of-delta-approximate}
The distortion of any {\sc Uniform-of-$\delta$-Approximate} mechanism is $O(k\delta)$.
\end{theorem}

\begin{proof}
Consider an arbitrary instance. Let $o$ be the optimal alternative, $a_d$ the representative of district $d$, and $w$ the final winner. Denote by $\SW_d(x)$ the social welfare of alternative $x$ only from the agents in $d$; clearly, $\SW(x) = \sum_{d \in D}\SW_d(x)$. The expected social welfare of the mechanism is
\begin{align*}
\mathbb{E}[\SW(M)] 
&= \sum_{a \in A} \Pr[w=a] \cdot \SW(a) \\
&= \frac{1}{k} \sum_{a \in A} \left( \sum_{d \in D} \Pr[a_d = a] \right) \SW(a) \\
&= \frac{1}{k} \sum_{d \in D} \sum_{a \in A} \Pr[a_d = a] \cdot \SW(a) \\
&= \frac{1}{k} \sum_{d \in D} \mathbb{E}[\SW(a_d)] \\
&\geq \frac{1}{k} \sum_{d \in D} \mathbb{E}[\SW_d(a_d)]
\end{align*}
Since $a_d$ is chosen based on a voting rule with distortion at most $\delta$, we have that $\mathbb{E}[\SW(a_d)] \geq \frac{1}{\delta} \cdot \SW_d(o)$. Combining this together with the fact that $\SW(o) = \sum_{d \in D} \SW_d(o)$, and using the linearity of expectation, we obtain
\begin{align*}
\mathbb{E}[\SW(M)] 
&\geq \frac{1}{k} \sum_{d \in D} \mathbb{E}[\SW_d(a_d)] \\
&\geq \frac{1}{k} \sum_{d \in D} \frac{1}{\delta}  \cdot \SW_d(o) \\
&= \frac{1}{k\delta} \cdot \SW(o).
\end{align*}
Hence, the distortion of the mechanism is at most $k\delta$.
\end{proof}

Theorem~\ref{thm:uniform-of-delta-approximate} is a simple composition theorem, analogous to the one presented by \citet{anshelevich2022distortion} for the metric setting. Based on it, we can define randomized distributed mechanisms with proven distortion guarantees by appropriately choosing the in-district rule. Before we continue, observe that the sizes of the districts do not appear in the proof of Theorem~\ref{thm:uniform-of-delta-approximate}, and thus the distortion of any {\sc Uniform-of-$\delta$-Approximate} mechanism is $O(k\delta)$ even if the districts are asymmetric. So, the distortion of the mechanism depends on the number of agents only if the distortion $\delta$ of the in-district rule depends on the number of agents. 

If cardinal information is available at the district level, by using {\sc Range-Voting} with $\delta=1$ as the in-district rule, we obtain the following. 

\begin{corollary} \label{cor:uniform-of-range-voting}
The distortion of {\sc Uniform-of-Range-Voting} is $O(k)$.
\end{corollary}

\noindent 
If only ordinal information about the preferences of the agents is given at the district level, then we can use {\sc Plurality} with $\delta = O(m^2)$ and the randomized rule {\sc Stable-Lottery} mechanism of \citet{ebadian2022optimized} with $\delta = O(\sqrt{m})$ as the in-district rule to obtain the following results.

\begin{corollary} \label{cor:uniform-of-plurality}
The distortion of {\sc Uniform-of-Plurality} is $O(k m^2)$.
\end{corollary}

\begin{corollary} \label{cor:uniform-of-stable-lottery}
The distortion of {\sc Uniform-of-Stable-Lottery} is $O(k \sqrt{m})$.
\end{corollary}

An important question to ask next is under what circumstances the aforementioned upper bounds of Corollaries~\ref{cor:uniform-of-range-voting}, \ref{cor:uniform-of-plurality} and \ref{cor:uniform-of-stable-lottery} are tight. First, we show that {\sc Uniform-of-Range-Voting} is the best among mechanisms with unanimous in-district rules which may even use cardinal information. 

\begin{theorem}
The distortion of any randomized distributed mechanism with a unanimous in-district rule is $\Omega(k)$. 
\end{theorem}

\begin{proof}
Let $\varepsilon > 0$ be a positive infinitesimal. 
Consider an instance with the following $k$ symmetric districts: For any $\ell \in [k]$, in district $d_\ell$, all $\lambda$ agents therein have value $1/2+\varepsilon$ for alternative $a_\ell$, $1/2-\varepsilon$ for alternative $x$, and $0$ for any other alternative. Since, the in-district rule is unanimous, the representative of district $d_\ell$ must be $a_\ell$ with probability $1$. Hence, no matter what the probability of choosing a district representative as the winner is, the expected social welfare of the mechanism is $\lambda \cdot (1/2+\varepsilon)$. However, the social welfare of alternative $x$ is $k\cdot \lambda \cdot (1/2 - \varepsilon)$, and thus the distortion is $\Omega(k)$.
\end{proof}

If we consider non-unanimous in-district rules, but require the in-district rule to be deterministic, then we can show a weaker lower bound of $\Omega(\sqrt{k})$; notice that the theorem also implies the same bound for fully deterministic distributed mechanisms without unanimous in-district rules.

\begin{theorem}
The distortion of any randomized distributed mechanism with a deterministic in-district rule is $\Omega(\sqrt{k})$. 
\end{theorem}

\begin{proof}
Consider a district $d_\ell$ in which all agents have value $1/2$ for alternative $a_\ell$, value $1/(2\sqrt{k})$ for each alternative in $\{b_1,\ldots,b_{\sqrt{k}}\}$, and $0$ for any other alternative. If the representative of this district is not $a_\ell$, then in instances consisting of copies of this district, the distortion is at least $\sqrt{k}$; in particular, it is at least that much if some alternative in $\{b_1,\ldots,b_{\sqrt{k}}\}$ is chosen and infinite if any other alternative is chosen. So, suppose that the representative of $d_\ell$ is $a_\ell$.

Next, consider an instance with $k$ symmetric districts $d_1, \ldots, d_k$. By the above discussion, for any $\ell \in [k]$, the representative of $d_\ell$ is alternative $a_\ell$ with social welfare $\lambda/2$ (note that only the agents of $d_\ell$ have positive value, equal to $1/2$, for $a_\ell$). Hence, no matter which district representative is chosen as the winner (or the probability distribution over the representatives), the (expected) social welfare of the mechanism is $\lambda/2$. In contrast, the social welfare of any alternative in $\{b_1,\ldots,b_{\sqrt{k}}\}$ is $k \cdot \lambda/(2\sqrt{k}) = \sqrt{k}\cdot \lambda/2$, and thus the distortion is $\sqrt{k}$.
\end{proof}

Next, we show that {\sc Uniform-of-Plurality} is the best possible among ordinal randomized distributed mechanisms with deterministic in-district rules, assuming an arbitrary but fixed ordering of the alternatives. 
This is quite surprising, as it shows that randomization over the districts is not better than just choosing an arbitrary alternative that is representative of the most districts (i.e., not better than {\sc Plurality-of-Plurality}).

\begin{theorem}\label{thm:ordinal-with-det-indistrict}
The distortion of any ordinal distributed mechanism with a deterministic in-district rule is $\Omega(km^2)$, when there exists an arbitrary but fixed tie-breaking ordering of the alternatives. 
\end{theorem}

\begin{proof}
Without loss of generality, suppose that the tie-breaking ordering of the alternatives is $a_1 \succ \ldots \succ  a_k \succ b_1 \succ \ldots \succ b_{m/2-1} \succ x \succ c_1 \succ \ldots \succ c_{m/2-k}$; the naming of the alternatives is arbitrary but is assumed to be known and can be exploited. 
For simplicity, for any set of alternatives $X$, denote by $[X]$ an arbitrary ordering of the alternatives in $X$. 

Consider an instance with $k$ symmetric districts such that in district $d_\ell$ there is 
a set of $2\lambda/m$ agents with preference ordering $a_\ell \succ x \succ [A \setminus \{a_\ell,x\}]$,
a set of $2\lambda/m$ agents with preference ordering $b_1 \succ x \succ [A \setminus \{b_1,x\}]$, 
$\ldots$,
and a set of $2\lambda/m$ agents with preference ordering $b_{m/2-1} \succ x \succ [A \setminus \{b_{m/2-1},x\}]$.
By Lemma~\ref{lem:top-or-infinite}, the representative of $d_\ell$ must be one of the top-ranked alternatives (otherwise the distortion of the mechanism would be infinite). Since $a_\ell$ is ranked above the other alternatives in the tie-breaking ordering, she chosen as the representative of $d_\ell$.
Hence, the set of representatives is $\{a_1, \ldots, a_k\}$, and the winner is chosen according to some probability distribution over this set.

The valuation profile may be such that the $2\lambda/m$ agents in district $d_\ell$ that rank $a_\ell$ first have value $1/m$ for all alternatives, while all other agents in $d_\ell$ have value $1/2$ for their two favorite alternatives. Consequently, the social welfare of alternative $a_\ell$ is $2\lambda/m^2$, and thus the social welfare of the mechanism is also this much, no matter the probability distribution over the district representatives. 
In contrast, the social welfare of $x$ is approximately $k\lambda/2$, leading to a distortion of $\Omega(km^2)$.
\end{proof}

When randomization at the district level can be leveraged by ordinal distributed mechanisms, then we achieve distortion much better than what is implied by Corollary~\ref{cor:uniform-of-stable-lottery}, while also achieving strategyproofness. In particular, there are several centralized voting rules that can be implemented as distributed mechanisms, in the sense that they define the \emph{same probability distribution} over the alternatives. One such important class of voting rules is that of {\em point-voting schemes}, which is part of a larger class of strategyproof mechanisms~\citep{Barbera1978,hylland1980strategy, gibbard1977manipulation} and includes rules with almost best possible distortion guarantees~\citep{boutilier2015optimal,ebadian2022optimized}. 

\subsection{Point-voting schemes}
A point-voting scheme chooses an agent uniformly at random and then outputs her $t$-th favorite alternative with probability $p_t$, where 
$p_1\geq \ldots \geq p_m \geq 0$ and $\sum_{t = 1}^m p_t = 1$. Hence, the probability according to which the point-voting scheme using the probability vector $\mathbf{p}=(p_1,\ldots,p_m)$ chooses alternative $a\in A$ as the winner $w$ is  $\Pr[w=a] = \frac{1}{n} \sum_{i \in N} p_{\sigma_i(a)}$, where $\sigma_i(a)$ is the position that $i$ ranks $a$ in her preference ranking $\sigma$.

There are many point-voting schemes of interest. 
For every positional scoring rule using the scoring vector $\mathbf{s} = (s_1,\ldots,s_m)$, we can define a point-voting scheme $f(\mathbf{s})$ by normalizing the scoring vector, that is, define $p_t = s_t / \left( \sum_{j \in [m]} s_j \right)$ for every $t \in [m]$ so that the winning probability of alternative $a$ is 
\begin{align*}
\Pr[w=a] 
&= \frac{1}{n} \sum_{i \in N} \frac{s_{\sigma_i(a)}}{\sum_{j \in [m]} s_j} \\
&= \frac{\sum_{i \in N}s_{\sigma_i(a)}}{n \cdot \sum_{j \in [m]} s_j}.
\end{align*}
Another important point-voting scheme is the rule that chooses each alternative uniformly at random; in this case, we have $p_t = 1/m$ for every $t \in [m]$ so that 
$\Pr[w=a] = \frac{1}{n} \sum_{i \in N} \frac{1}{m} = \frac{1}{m}$. 

For any point-voting scheme $f$ that uses a probability vector $\mathbf{p}$, we consider the distributed mechanism {\sc Proportional-of-$f$-Point-Voting}, which works as follows:
\begin{itemize}
    \item For every district $d$, choose the representative $a_d$ to be alternative $a \in A$ with probability $\frac{1}{\lambda} \sum_{i \in N_d} p_{\sigma_i(a)}$.
    \item Choose the winner to be the representative of district $d$ with probability $n_d/n$.
\end{itemize}

\begin{theorem} \label{thm:distributed-implementation-of-point-voting-scheme}
{\sc Proportional-of-$f$-Point-Voting} defines the same probability distribution as the point-voting scheme $f$. 
\end{theorem}

\begin{proof}
The probability that alternative $a$ is chosen as the winner by {\sc Proportional-of-$f$-Point-Voting} is
\begin{align*}
\Pr[w = a] 
&= \sum_{d \in D} \Pr[w = a_d] \cdot \Pr[a_d = a] \\
&= \sum_{d \in D} \frac{n_d}{n} \cdot \frac{1}{n_d} \sum_{i \in N_d} p_{\sigma_i(a)} \\
&= \frac{1}{n} \sum_{i \in N} p_{\sigma_i(a)},
\end{align*}
that is, {\sc Proportional-of-$f$-Point-Voting} chooses $a$ with the same probability as $f$. 
\end{proof}

Theorem~\ref{thm:distributed-implementation-of-point-voting-scheme} shows that {\sc Proportional-of-$f$-Point-Voting} achieves the same distortion bound as the point-voting scheme $f$ it uses as the in-district rule, and also that it inherits its strategyproofness property. This is extremely useful, as there are centralized voting rules that are based on point-voting schemes and achieve almost the best possible distortion. 

\citet{boutilier2015optimal} considered a voting rule that is a convex combination of two point-voting schemes: With probability $1/2$ choose an alternative uniformly at random, and with probability $1/2$ run the point-voting scheme defined by normalizing the harmonic scoring rule $H = (1, 1/2, \ldots, 1/m)$. We will refer to this mechanism as {\sc BCHLPS}. \citet{boutilier2015optimal} showed that this voting rule has distortion $O(\sqrt{m \log{m}})$. An important property of point-voting schemes is that any rule that is a convex combination of point-voting schemes is also a point-voting scheme. The following lemma is similar to lemmas proved before in the literature (e.g., see \cite{filos2014truthful,Barbera1978}); we provide a proof for completeness.

\begin{lemma}\label{lem:combination-of-point-voting-schemes}
Let $f_1,\ldots,f_\kappa$ be point-voting schemes defined by the probability vectors $\mathbf{p}_1, \ldots, \mathbf{p}_\kappa$. 
For any non-negative numbers $q_1,\ldots,q_\kappa$ such that $\sum_{j \in [\kappa]} q_j = 1$, the voting rule $f$ that chooses the outcome of $f_j$ with probability $q_j$ is a point-voting scheme.
\end{lemma}

\begin{proof}
Let $\sigma$ be an arbitrary preference profile. 
For any $j \in [\kappa]$, denote the $t$-th coordinate of $\mathbf{p}_j$ as $p_{j,t}$, and let $P_j(a) = \Pr[a = f_j(\sigma)]$ be the probability of choosing $a$ as the winner according to point-voting scheme $f_j$. Then, the voting rule $f$ chooses alternative $a$ as the winner $w$ with probability 
\begin{align*}
\Pr[w=a] 
&= \sum_{j \in [\kappa]} q_j \cdot P_j(a) \\
&= \sum_{j \in [\kappa]} q_j \cdot \left( \frac{1}{n} \sum_{i \in N} p_{j,\sigma_i(a)} \right) \\
&= \frac{1}{n} \sum_{i \in N} \sum_{j \in [\kappa]} q_j \cdot p_{j,\sigma_i(a)}.
\end{align*}
Hence, $f$ is a point-voting scheme defined by the probability vector $\mathbf{p}$ with $p_t = \sum_{j \in [\kappa]} q_j \cdot p_{j,t}$.
\end{proof}

Consequently, by Theorem~\ref{thm:distributed-implementation-of-point-voting-scheme} and Lemma~\ref{lem:combination-of-point-voting-schemes}, we can construct a randomized ordinal distributed mechanism based on the point-voting scheme of \citet{boutilier2015optimal} that achieves the same distortion bound and is strategyproof.

\begin{corollary} \label{cor:BoutilierBest}
There exists a randomized ordinal strategyproof distributed mechanism with distortion $O(\sqrt{m \log{m}})$.
\end{corollary}

This distortion bound is almost best possible as the lower bound of $\Omega(\sqrt{m})$ for randomized centralized rules holds trivially for distributed mechanisms by considering single-district instances.


\afterpage{%
    \clearpage
    \thispagestyle{empty}
    \begin{landscape}

\begin{table}[!th]
\centering
\small
\begin{tabular}{|c|c|c|c|c|c|c|c|c|c|c|}
\hline
   $k$  & {\sc RV} & {\sc PL} & {\sc Veto}  & {\sc Borda} & {\sc Harmonic} & {\sc PropPL} & {\sc PropVeto} & {\sc PropBorda} & {\sc PropHarmonic} & {\sc BCHLPS} \\ \hline
$1$  & 1            & 1.049     & 1.035 & 1.007 & 1.017    & 1.135         & 1.166    & 1.155     & 1.156        & 1.166  \\ \hline
$2$  & 1.017        & 1.070     & 1.059 & 1.018 & 1.020    & 1.137         & 1.166    & 1.155     & 1.156        & 1.165  \\ \hline
$5$  & 1.018        & 1.064     & 1.070 & 1.020 & 1.036    & 1.133         & 1.162    & 1.155     & 1.156        & 1.165  \\ \hline
$10$ & 1.019        & 1.066     & 1.082 & 1.021 & 1.044    & 1.133         & 1.162    & 1.153     & 1.154        & 1.163  \\ \hline
$20$ & 1.024        & 1.066     & 1.107 & 1.030 & 1.050    & 1.134         & 1.165    & 1.154     & 1.155        & 1.164  \\ \hline
$25$ & 1.022        & 1.067     & 1.142 & 1.031 & 1.107    & 1.133         & 1.165    & 1.153     & 1.154        & 1.164  \\ \hline
\end{tabular}
\caption{Distortion bounds of various voting rules based on valuations defined by the provided scores of the Jester dataset and random district partitions.}
\label{tab:RandomDistrictsJester}
\end{table}

\begin{table}[!ht]
\small
\centering
\begin{tabular}{|l|c|c|c|c|c!{\vrule width 0.5mm}c|c|c|c|c|c|}
\hline
& {\sc RV} & {\sc PL} & {\sc Veto}  & {\sc Borda} & {\sc Harmonic} & {\sc PropPL} & {\sc PropVeto} & {\sc PropBorda} & {\sc PropHarmonic} & {\sc BCHLPS} \\
\hline
\rowcolor{gray!30!white}
$k=1$  &              &           &       &       &          &               &          &           &              &          \\
Uniform     & 1            & 1.038     & 1.045 & 1.006 & 1.019    & 1.079         & 1.087    & 1.085     & 1.085        & 1.087  \\
Beta        & 1            & 1.086     & 1.105 & 1.029 & 1.050    & 1.140         & 1.152    & 1.147     & 1.147        & 1.150  \\
Exponential & 1            & 1.032     & 1.096 & 1.018 & 1.013    & 1.118         & 1.137    & 1.132     & 1.131        & 1.134  \\
\hline
\rowcolor{gray!30!white}
$k=2$         &              &           &       &       &          &               &          &           &              &        \\
Uniform     & 1.026        & 1.052     & 1.056 & 1.030 & 1.039    & 1.079         & 1.087    & 1.084     & 1.084        & 1.086  \\
Beta        & 1.044        & 1.111     & 1.118 & 1.064 & 1.080    & 1.140         & 1.152    & 1.147     & 1.147        & 1.150  \\
Exponential & 1.039        & 1.062     & 1.115 & 1.055 & 1.051    & 1.118         & 1.136    & 1.132     & 1.130        & 1.135  \\
\hline
\rowcolor{gray!30!white}
$k=5$         &              &           &       &       &          &               &          &           &              &        \\
Uniform     & 1.031        & 1.050     & 1.057 & 1.029 & 1.038    & 1.076         & 1.084    & 1.081     & 1.081        & 1.084  \\
Beta        & 1.052        & 1.113     & 1.125 & 1.074 & 1.094    & 1.143         & 1.155    & 1.151     & 1.150        & 1.154  \\
Exponential & 1.039        & 1.069     & 1.110 & 1.055 & 1.056    & 1.119         & 1.137    & 1.133     & 1.131        & 1.134  \\
\hline
\rowcolor{gray!30!white}
$k=20$        &              &           &       &       &          &               &          &           &              &        \\
Uniform     & 1.031        & 1.055     & 1.077 & 1.039 & 1.042    & 1.077         & 1.085    & 1.082     & 1.082        & 1.084  \\
Beta        & 1.055        & 1.105     & 1.145 & 1.073 & 1.084    & 1.141         & 1.154    & 1.149     & 1.149        & 1.152  \\
Exponential & 1.047        & 1.069     & 1.123 & 1.060 & 1.058    & 1.115         & 1.133    & 1.128     & 1.127        & 1.129  \\
\hline
\rowcolor{gray!30!white}
$k=25$        &              &           &       &       &          &               &          &           &              &        \\
Uniform     & 1.031        & 1.056     & 1.071 & 1.036 & 1.044    & 1.077         & 1.085    & 1.082     & 1.0824       & 1.084  \\
Beta        & 1.054        & 1.124     & 1.149 & 1.084 & 1.094    & 1.148         & 1.155    & 1.150     & 1.150        & 1.151  \\
Exponential & 1.042        & 1.069     & 1.129 & 1.060 & 1.054    & 1.116         & 1.134    & 1.129     & 1.128        & 1.131  \\
\hline 
\end{tabular}
\caption{Distortion bounds of various voting rules based on valuations defined according to several probability distributions and random district partitions. Results for deterministic mechanisms are presented at the left of the bold vertical line, and results for randomized mechanisms are at the right of the bold vertical line. }
\label{tab:AllRandom}
\end{table}

    \end{landscape}
    \clearpage
}

\section{Experiments} \label{sec:experiments}
In this section, we perform experiments with real and synthetic datasets, aiming to identify patterns in the distortion of several well-known voting rules and examine whether these support our theoretical findings. It is well-documented in the literature (e.g., see \citep{boutilier2015optimal, distributed2020}) that when working with real or realistic preferences, it often is the case that the distortions bounds are small numbers quite close to $1$. In this sense, our goal is not primarily to demonstrate the distortion bounds themselves, but rather the dependence of these bounds on the distributed decision-making process, in particular the number of districts, as well as the use of randomization. We perform two main experiments, one with real-world preferences and valuation data, and one with synthetic data. All our experiments are with symmetric districts.

\subsection{Experiments with the Jester Dataset} 

For our first experiment, we use the Jester Joke Dataset \citep{goldberg2001eigentaste}. The dataset contains ratings for $100$ different jokes in the range $[-10,10]$, provided by $70000$ users. We chose to work with this dataset as it has also been employed by \citet{boutilier2015optimal} in the context of centralized distortion bounds, and also by \citet{distributed2020} for the distortion of deterministic distributed mechanisms that use plurality as the over-district rule.

Following the methodology developed in these works, we construct inputs consisting of ratings for the $8$ most-rated jokes. In particular, we perform $1000$ random runs in which we sample $100$ users from the set of all users that have provided rankings for all eight jokes, and then partition them into $k$ equal-sized districts uniformly at random, for $k \in \{1, 2, 5, 10, 20, 25\}$. Clearly, the case of $k=1$ corresponds to the centralized setting and will be used as a reference point. We interpret the ratings of the jokes as cardinal valuations: to be consistent with our setting (and with the experiments of \citep{boutilier2015optimal,distributed2020}), we add $10$ to each user's rating vector, to ensure that the values are positive and then apply the unit-sum normalization. 
For these inputs, we compute the average distortion of a set of $20$ voting rules over the $1000$ runs of the experiment. In particular, we consider distributed mechanisms $f_\text{over}$-of-$f_\text{in}$, where for $f_\text{over}$ we use \textsc{Plurality} or \textsc{Uniform}, whereas for $f_\text{in}$ we have: \medskip

\noindent \textbf{Deterministic Rules:} We use simple voting scoring rules, namely \textsc{Plurality (PL)}, \textsc{Veto}, \textsc{Borda} and \textsc{Harmonic}, 
as well as \textsc{Range-Voting (RV)}, which in the case of $k=1$ finds the optimal alternative. \smallskip

 \noindent \textbf{Randomized Rules:} Here we use several natural point-voting schemes with probability vectors that are proportional to the aforementioned scoring rules (recall the definition from Section~\ref{sec:randomized}), namely
 \begin{itemize}
     \item \textsc{Proportional to Plurality Score (PropPL)};
     \item \textsc{Proportional to Borda Score (PropBorda)};
     \item \textsc{Proportional to Veto Score (PropVeto)};
     \item \textsc{Proportional to Harmonic Score (PropHarmonic)}.
 \end{itemize}
 We also use the rule of \citet{boutilier2015optimal} (we refer to it as BCHLPS in the following); recall that this is a point-voting scheme that with probability $1/2$ selects an alternative at random and with probability $1/2$ runs the PropHarmonic rule defined above. As established in Corollary~\ref{cor:BoutilierBest} (and the discussion before the statement of the corollary), this is best possible in terms of the worst-case distortion. 
 
The results of our experiments can be seen in Table~\ref{tab:RandomDistrictsJester}. In the table we only present the results where as $f_{\text{over}}$, we used \textsc{Plurality} for deterministic rules and \textsc{Uniform} for randomized rules. This is in accordance to our approach in the theoretical results in previous sections. The bounds for the cases not shown are quite similar, and slightly larger in general. For each of the randomized rules, we perform $300$ runs and calculate their expected social welfare, which we then use to calculate the distortion.

From the results of Table~\ref{tab:RandomDistrictsJester} we observe that, as expected, the existence of multiple districts has an adverse effect on the distortion of deterministic mechanisms, which becomes worse compared to the centralized case $k=1$. For these rules, we can also observe that the distortion generally increases as $k$ increases. In contrast, the distortion of randomized rules remains virtually unchanged for any value of $k$. This is in complete accordance with our theoretical findings, where we established that these rules induce the same probability distribution. The experiments showcase that this does not only hold in expectation, but also in practice (given sufficiently many runs). 

Another crucial observation is that, in terms of the absolute distortion numbers, randomization does not seem to help; if anything, it makes the distortion bounds worse! This can be justified by the fact that real-world instances like those from the Jester dataset display a large degree of homogeneity, which results in the simple deterministic rules performing quite well. On the other hand, randomization often leads to suboptimal choices even on such ``well-behaved'' instances, demeaning the distortion bounds on average. Surprisingly, among ordinal voting rules, \textsc{Borda} seems to perform best across the board even though the theoretical distortion of \textsc{Borda} is in fact unbounded.  

\subsection{Experiments with Synthetic Datasets}

We also perform experiments with datasets that are generated from probability distributions. In particular, and to be consistent with the Jester experiment presented above, we create instances with $100$ agents and $8$ alternatives, by first drawing the values of the agents from a certain distribution, and then constructing the induced ordinal preference profile from those values. We use the following distributions:
\begin{itemize}
    \item {\em Uniform distribution} in $[1,100]$. This is the simplest case, where all possible values are equally likely.
    \item {\em Beta distribution} with $\alpha = 1/10$ and $\beta = 1/10$. This distribution has a symmetric convex pdf function centered around a mean of $1/2$, assigning higher probabilities to values very close to $1$ or $0$. 
    \item {\em Exponential distribution} with exponent $4$, i.e., the pdf is $f(x)=4e^{4}$ for $x \geq 0$ and $f(x)=0$ otherwise. This distribution generates values close to $0$ with high probability, and as the values increase, the probability of them being generated decreases exponentially. 
\end{itemize}
For the rest of the experiment, we perform similar steps as in the case of the Jester dataset: 
We normalize the values to sum up to $1$, and run the set of mechanisms described above. For each randomized mechanism we now perform $150$ individual runs and calculate its expected welfare. We calculate the average distortions over $500$ runs of the experiment for $k$ symmetric districts, where $k \in \{1, 2, 5, 20, 25\}$. Note that the number of runs and the number of district sizes is slightly smaller in this experiment, because it is more computationally intensive (as we need to calculate bounds for $3$ different distributions). Again, we use \textsc{Plurality} as $f_{\text{over}}$ for deterministic and \textsc{Uniform} for randomized mechanisms; the results for the other cases were similar and are not reported. 

The results can be found in Table~\ref{tab:AllRandom}. Similarly to the Jester experiment, it is evident that the distortion of the deterministic mechanisms becomes worse for $k \geq 2$, whereas it remains pretty much the same for randomized mechanisms. Again, we observe that randomization results in worse distortion bounds overall, and that \textsc{Borda} performs best among deterministic mechanisms. Interestingly, contrary to the Jester dataset, here we do not see a clear pattern of the distortion increasing as $k$ increases for deterministic mechanisms (other than the jump from $k=1$ to $k=2$). This is probably due to the fact that the synthetic instances are highly homogeneous, and with uniform random district partitions, the districts end up being quite uniform, regardless of their number and size. 

\smallskip
\noindent
{\em The role of unit-sum.} 
We remark here that normalizing the values to sum up to $1$ effectively makes the Uniform and Exponential 
distributions pretty similar, and this is reflected in the results. To get a sense of the effect of normalization, we also ran the experiments without it. We observe that the distortions for the exponential 
distribution are now larger than those of the uniform distribution. In general, the distortion bounds still lie in the range $[1.03,1.15]$ for all distributions, but their average values (over all documented distortion bounds) are larger for all distributions except Uniform. It is also the case that for the Beta distribution, the bounds of deterministic mechanisms are much closer to those of randomized ones. The distortion of randomized mechanisms is still almost the same for any number of districts.  

\section{Open Problems}

From our results, an interesting technical challenge is to remove the requirement for a consistent tie-breaking ordering from the statement of Theorem~\ref{thm:ordinal-with-det-indistrict}. Similarly, we could attempt to remove unanimity from the lower bound of Theorem~\ref{thm:dist-unanimous}; although unanimity is usually pretty natural, removing it would make the theorem stronger. 
More interestingly, our result about point-voting schemes in Theorem~\ref{thm:distributed-implementation-of-point-voting-scheme} crucially does \emph{not} depend on the normalization of the valuations, and hence also could be applied verbatim to the metric distributed social choice setting studied by \citet{anshelevich2022distortion}, where randomized mechanisms have never been considered; this seems like a natural starting point for such an investigation.

\bibliographystyle{plainnat} 
\bibliography{references}

\begin{thebibliography}{35}
\providecommand{\natexlab}[1]{#1}
\providecommand{\url}[1]{\texttt{#1}}
\expandafter\ifx\csname urlstyle\endcsname\relax
  \providecommand{\doi}[1]{doi: #1}\else
  \providecommand{\doi}{doi: \begingroup \urlstyle{rm}\Url}\fi

\bibitem[Abramowitz et~al.(2019)Abramowitz, Anshelevich, and
  Zhu]{abramowitz2019awareness}
Ben Abramowitz, Elliot Anshelevich, and Wennan Zhu.
\newblock Awareness of voter passion greatly improves the distortion of metric
  social choice.
\newblock In \emph{Proceedings of the The 15th Conference on Web and Internet
  Economics (WINE)}, pages 3--16, 2019.

\bibitem[Amanatidis et~al.(2021)Amanatidis, Birmpas, Filos{-}Ratsikas, and
  Voudouris]{Amanatidis2020peeking}
Georgios Amanatidis, Georgios Birmpas, Aris Filos{-}Ratsikas, and Alexandros~A.
  Voudouris.
\newblock Peeking behind the ordinal curtain: Improving distortion via cardinal
  queries.
\newblock \emph{Artificial Intelligence}, 296:\penalty0 103488, 2021.

\bibitem[Amanatidis et~al.(2022{\natexlab{a}})Amanatidis, Birmpas,
  Filos{-}Ratsikas, and Voudouris]{amanatidis2021matching}
Georgios Amanatidis, Georgios Birmpas, Aris Filos{-}Ratsikas, and Alexandros~A.
  Voudouris.
\newblock A few queries go a long way: Information-distortion tradeoffs in
  matching.
\newblock \emph{Journal of Artificial Intelligence Research}, 74,
  2022{\natexlab{a}}.

\bibitem[Amanatidis et~al.(2022{\natexlab{b}})Amanatidis, Birmpas,
  Filos-Ratsikas, and Voudouris]{amanatidis2022don}
Georgios Amanatidis, Georgios Birmpas, Aris Filos-Ratsikas, and Alexandros~A.
  Voudouris.
\newblock Don't roll the dice, ask twice: The two-query distortion of matching
  problems and beyond.
\newblock In \emph{Proceedings of the 36th Conference on Neural Information
  Processing Systems ({NeurIPS})}, 2022{\natexlab{b}}.

\bibitem[Anshelevich and Sekar(2016)]{anshelevich2016blind}
Elliot Anshelevich and Shreyas Sekar.
\newblock Blind, greedy, and random: Algorithms for matching and clustering
  using only ordinal information.
\newblock In \emph{Proceedings of the 30th {AAAI} Conference on Artificial
  Intelligence ({AAAI})}, pages 390--396, 2016.

\bibitem[Anshelevich et~al.(2018)Anshelevich, Bhardwaj, Elkind, Postl, and
  Skowron]{anshelevich2018approximating}
Elliot Anshelevich, Onkar Bhardwaj, Edith Elkind, John Postl, and Piotr
  Skowron.
\newblock Approximating optimal social choice under metric preferences.
\newblock \emph{Artificial Intelligence}, 264:\penalty0 27--51, 2018.

\bibitem[Anshelevich et~al.(2021)Anshelevich, Filos-Ratsikas, Shah, and
  Voudouris]{survey2021}
Elliot Anshelevich, Aris Filos-Ratsikas, Nisarg Shah, and Alexandros~A.
  Voudouris.
\newblock Distortion in social choice problems: The first 15 years and beyond.
\newblock In \emph{Proceedings of the 30th International Joint Conference on
  Artificial Intelligence {(IJCAI)}}, pages 4294--4301, 2021.

\bibitem[Anshelevich et~al.(2022)Anshelevich, Filos-Ratsikas, and
  Voudouris]{anshelevich2022distortion}
Elliot Anshelevich, Aris Filos-Ratsikas, and Alexandros~A Voudouris.
\newblock The distortion of distributed metric social choice.
\newblock \emph{Artificial Intelligence}, 308:\penalty0 103713, 2022.

\bibitem[Bachrach et~al.(2016)Bachrach, Lev, Lewenberg, and
  Zick]{bachrach2016misrepresentation}
Yoram Bachrach, Omer Lev, Yoad Lewenberg, and Yair Zick.
\newblock Misrepresentation in district voting.
\newblock In \emph{IJCAI}, pages 81--87, 2016.

\bibitem[Barbera(1978)]{Barbera1978}
Salvador Barbera.
\newblock \emph{Nice Decision Schemes}.
\newblock Springer Netherlands, 1978.

\bibitem[Benad{\`{e}} et~al.(2017)Benad{\`{e}}, Nath, Procaccia, and
  Shah]{benade2017participatory}
Gerdus Benad{\`{e}}, Swaprava Nath, Ariel~D. Procaccia, and Nisarg Shah.
\newblock Preference elicitation for participatory budgeting.
\newblock In \emph{Proceedings of the 31st {AAAI} Conference on Artificial
  Intelligence ({AAAI})}, pages 376--382, 2017.

\bibitem[Bhaskar and Ghosh(2018)]{bhaskar2018cardinal}
Umang Bhaskar and Abheek Ghosh.
\newblock On the welfare of cardinal voting mechanisms.
\newblock In \emph{Proceedings of the 38th {IARCS} Annual Conference on
  Foundations of Software Technology and Theoretical Computer Science
  ({FSTTCS})}, pages 27:1--27:22, 2018.

\bibitem[Bhaskar et~al.(2018)Bhaskar, Dani, and
  Ghosh]{bhaskar2018truthful-multiwinner}
Umang Bhaskar, Varsha Dani, and Abheek Ghosh.
\newblock Truthful and near-optimal mechanisms for welfare maximization in
  multi-winner elections.
\newblock In \emph{Proceedings of the 32nd {AAAI} Conference on Artificial
  Intelligence ({AAAI})}, pages 925--932, 2018.

\bibitem[Borodin et~al.(2018)Borodin, Lev, Shah, and Strangway]{borodin2018big}
Allan Borodin, Omer Lev, Nisarg Shah, and Tyrone Strangway.
\newblock Big city vs. the great outdoors: Voter distribution and how it
  affects gerrymandering.
\newblock In \emph{IJCAI}, pages 98--104, 2018.

\bibitem[Borodin et~al.(2019)Borodin, Lev, Shah, and
  Strangway]{borodin2019primaries}
Allan Borodin, Omer Lev, Nisarg Shah, and Tyrone Strangway.
\newblock Primarily about primaries.
\newblock In \emph{Proceedings of the 33rd {AAAI} Conference on Artificial
  Intelligence ({AAAI})}, pages 1804--1811, 2019.

\bibitem[Boutilier et~al.(2015)Boutilier, Caragiannis, Haber, Lu, Procaccia,
  and Sheffet]{boutilier2015optimal}
Craig Boutilier, Ioannis Caragiannis, Simi Haber, Tyler Lu, Ariel~D. Procaccia,
  and Or~Sheffet.
\newblock Optimal social choice functions: A utilitarian view.
\newblock \emph{Artificial Intelligence}, 227:\penalty0 190--213, 2015.

\bibitem[Brandt et~al.(2016)Brandt, Conitzer, Endriss, Lang, and
  Procaccia]{comsoc-book}
Felix Brandt, Vincent Conitzer, Ulle Endriss, J{\'{e}}r{\^{o}}me Lang, and
  Ariel~D. Procaccia, editors.
\newblock \emph{Handbook of Computational Social Choice}.
\newblock Cambridge University Press, 2016.

\bibitem[Caragiannis et~al.(2017)Caragiannis, Nath, Procaccia, and
  Shah]{caragiannis2017subset}
Ioannis Caragiannis, Swaprava Nath, Ariel~D. Procaccia, and Nisarg Shah.
\newblock Subset selection via implicit utilitarian voting.
\newblock \emph{Journal of Artificial Intelligence Research}, 58:\penalty0
  123--152, 2017.

\bibitem[Caragiannis et~al.(2022)Caragiannis, Shah, and Voudouris]{CSV22}
Ioannis Caragiannis, Nisarg Shah, and Alexandros~A. Voudouris.
\newblock The metric distortion of multiwinner voting.
\newblock In \emph{Proceedings of the 36th {AAAI} Conference on Artificial
  Intelligence ({AAAI})}, pages 4900--4907, 2022.

\bibitem[Ebadian et~al.(2022)Ebadian, Kahng, Peters, and
  Shah]{ebadian2022optimized}
Soroush Ebadian, Anson Kahng, Dominik Peters, and Nisarg Shah.
\newblock Optimized distortion and proportional fairness in voting.
\newblock In \emph{Proceedings of the 23rd {ACM} Conference on Economics and
  Computation ({EC})}, pages 563--600, 2022.

\bibitem[Elkind et~al.(2021)Elkind, Gan, Obraztsova, Rabinovich, and
  Voudouris]{elkind2021protecting}
Edith Elkind, Jiarui Gan, Svetlana Obraztsova, Zinovi Rabinovich, and
  Alexandros~A. Voudouris.
\newblock Protecting elections by recounting ballots.
\newblock \emph{Artificial Intelligence}, 290:\penalty0 103401, 2021.

\bibitem[Filos-Ratsikas and Miltersen(2014)]{filos2014truthful}
Aris Filos-Ratsikas and Peter~Bro Miltersen.
\newblock Truthful approximations to range voting.
\newblock In \emph{Proceedings of the 10th International Conference on Web and
  Internet Economics ({WINE})}, pages 175--188, 2014.

\bibitem[Filos-Ratsikas and Voudouris(2021)]{filos2021approximate}
Aris Filos-Ratsikas and Alexandros~A Voudouris.
\newblock Approximate mechanism design for distributed facility location.
\newblock In \emph{International Symposium on Algorithmic Game Theory}, pages
  49--63. Springer, 2021.

\bibitem[Filos-Ratsikas et~al.(2014)Filos-Ratsikas, Frederiksen, and
  Zhang]{Aris14}
Aris Filos-Ratsikas, S{\o}ren Kristoffer~Stiil Frederiksen, and Jie Zhang.
\newblock {Social welfare in one-sided matchings: Random priority and beyond}.
\newblock In \emph{Proceedings of the 7th Symposium of Algorithmic Game Theory
  ({SAGT})}, pages 1--12, 2014.

\bibitem[Filos{-}Ratsikas et~al.(2020)Filos{-}Ratsikas, Micha, and
  Voudouris]{distributed2020}
Aris Filos{-}Ratsikas, Evi Micha, and Alexandros~A. Voudouris.
\newblock The distortion of distributed voting.
\newblock \emph{Artificial Intelligence}, 286:\penalty0 103343, 2020.

\bibitem[Gibbard(1977)]{gibbard1977manipulation}
Allan Gibbard.
\newblock Manipulation of schemes that mix voting with chance.
\newblock \emph{Econometrica: Journal of the Econometric Society}, pages
  665--681, 1977.

\bibitem[Gkatzelis et~al.(2020)Gkatzelis, Halpern, and
  Shah]{gkatzelis2020resolving}
Vasilis Gkatzelis, Daniel Halpern, and Nisarg Shah.
\newblock Resolving the optimal metric distortion conjecture.
\newblock In \emph{Proceedings of the 61st {IEEE} Annual Symposium on
  Foundations of Computer Science ({FOCS})}, pages 1427--1438, 2020.

\bibitem[Goldberg et~al.(2001)Goldberg, Roeder, Gupta, and
  Perkins]{goldberg2001eigentaste}
Ken Goldberg, Theresa Roeder, Dhruv Gupta, and Chris Perkins.
\newblock Eigentaste: A constant time collaborative filtering algorithm.
\newblock \emph{information retrieval}, 4\penalty0 (2):\penalty0 133--151,
  2001.

\bibitem[Hylland(1980)]{hylland1980strategy}
Aanund Hylland.
\newblock Strategy proofness of voting procedures with lotteries as outcomes
  and infinite sets of strategies.
\newblock Technical report, 1980.

\bibitem[Kizilkaya and Kempe(2022)]{kempe22veto}
Fatih~Erdem Kizilkaya and David Kempe.
\newblock Plurality veto: {A} simple voting rule achieving optimal metric
  distortion.
\newblock In \emph{Proceedings of the 31st International Joint Conference on
  Artificial Intelligence ({IJCAI})}, pages 349--355, 2022.

\bibitem[Lev and Lewenberg(2019)]{lev2019reverse}
Omer Lev and Yoad Lewenberg.
\newblock “reverse gerrymandering”: Manipulation in multi-group decision
  making.
\newblock In \emph{Proceedings of the AAAI Conference on Artificial
  Intelligence}, volume~33, pages 2069--2076, 2019.

\bibitem[Lewenberg et~al.(2017)Lewenberg, Lev, and
  Rosenschein]{lewenberg2017divide}
Yoad Lewenberg, Omer Lev, and Jeffrey~S Rosenschein.
\newblock Divide and conquer: Using geographic manipulation to win
  district-based elections.
\newblock In \emph{Proceedings of the 16th Conference on Autonomous Agents and
  MultiAgent Systems}, pages 624--632, 2017.

\bibitem[Mandal et~al.(2019)Mandal, Procaccia, Shah, and
  Woodruff]{mandal2019thrifty}
Debmalya Mandal, Ariel~D. Procaccia, Nisarg Shah, and David~P. Woodruff.
\newblock Efficient and thrifty voting by any means necessary.
\newblock In \emph{Proceedings of the 32nd Annual Conference on Neural
  Information Processing Systems ({NeurIPS})}, pages 7178--7189, 2019.

\bibitem[Mandal et~al.(2020)Mandal, Shah, and Woodruff]{mandal2020optimal}
Debmalya Mandal, Nisarg Shah, and David~P. Woodruff.
\newblock Optimal communication-distortion tradeoff in voting.
\newblock In \emph{Proceedings of the 21st {ACM} Conference on Economics and
  Computation ({EC})}, pages 795--813, 2020.

\bibitem[Procaccia and Rosenschein(2006)]{procaccia2006distortion}
Ariel~D. Procaccia and Jeffrey~S. Rosenschein.
\newblock The distortion of cardinal preferences in voting.
\newblock In \emph{International Workshop on Cooperative Information Agents
  ({CIA})}, pages 317--331, 2006.

\end{thebibliography}

\end{document}